\documentclass{article}
\usepackage{graphicx} 
\usepackage{amssymb}
\usepackage{color}
\usepackage{amsmath}
\usepackage{authblk}
\usepackage{float}
\usepackage{subcaption}

\usepackage{amsfonts,caption,epsfig,graphics,hyperref,latexsym,mathrsfs,revsymb,theorem,url,verbatim,epstopdf,mathtools,enumerate}

\newtheorem{definition}{Definition}
\newtheorem{question}[definition]{Question}
\newtheorem{lemma}[definition]{Lemma}
\newtheorem{remark}[definition]{Remark}
\newtheorem{theorem}[definition]{Theorem}
\newtheorem{example}[definition]{Example}
\newtheorem{proposition}[definition]{Proposition}

\newtheorem{corollary}[definition]{Corollary}
\newtheorem{conjecture}[definition]{Conjecture}

\newtheorem{memo}[definition]{Memo}

\def\squareforqed{\hbox{\rlap{$\sqcap$}$\sqcup$}}
\def\qed{\ifmmode\squareforqed\else{\unskip\nobreak\hfil
\penalty50\hskip1em\null\nobreak\hfil\squareforqed
\parfillskip=0pt\finalhyphendemerits=0\endgraf}\fi}
\def\endenv{\ifmmode\;\else{\unskip\nobreak\hfil
\penalty50\hskip1em\null\nobreak\hfil\;
\parfillskip=0pt\finalhyphendemerits=0\endgraf}\fi}
\newenvironment{proof}{\noindent \textbf{{Proof.~} }}{\qed}
\def\Dbar{\leavevmode\lower.6ex\hbox to 0pt
{\hskip-.23ex\accent"16\hss}D}
\makeatletter
\def\url@leostyle{%
  \@ifundefined{selectfont}{\def\UrlFont{\sf}}{\def\UrlFont{\small\ttfamily}}}
\makeatother

\def\bcj{\begin{conjecture}}
\def\ecj{\end{conjecture}}
\def\bcr{\begin{corollary}}
\def\ecr{\end{corollary}}
\def\bd{\begin{definition}}
\def\ed{\end{definition}}
\def\bea{\begin{eqnarray}}
\def\eea{\end{eqnarray}}
\def\beq{\begin{equation}}
\def\eeq{\end{equation}}
\def\bal{\begin{aligned}}
\def\eal{\end{aligned}}
\def\bem{\begin{enumerate}}
\def\eem{\end{enumerate}}
\def\bex{\begin{example}}
\def\eex{\end{example}}
\def\bim{\begin{itemize}}
\def\eim{\end{itemize}}
\def\bl{\begin{lemma}}
\def\el{\end{lemma}}
\def\bma{\begin{bmatrix}}
\def\ema{\end{bmatrix}}
\def\bpf{\begin{proof}}
\def\epf{\end{proof}}
\def\bpp{\begin{proposition}}
\def\epp{\end{proposition}}
\def\bqu{\begin{question}}
\def\equ{\end{question}}
\def\br{\begin{remark}}
\def\er{\end{remark}}
\def\bt{\begin{theorem}}
\def\et{\end{theorem}}
\def\bmm{\begin{memo}}
\def\emm{\end{memo}}

\def\btb{\begin{tabular}}
\def\etb{\end{tabular}}

\newcommand{\nc}{\newcommand}

\def\a{\alpha}
\def\b{\beta}
\def\g{\gamma}

\def\l{\lambda}

\def\r{\rho}
\def\s{\sigma}

\def\ps{\psi}

\def\G{\Gamma}

\def\S{\Sigma}

\nc{\bbA}{\mathbb{A}} \nc{\bbB}{\mathbb{B}} \nc{\bbC}{\mathbb{C}}
 \nc{\bbD}{\mathbb{D}} \nc{\bbE}{\mathbb{E}} \nc{\bbF}{\mathbb{F}}
 \nc{\bbG}{\mathbb{G}} \nc{\bbH}{\mathbb{H}} \nc{\bbI}{\mathbb{I}}
 \nc{\bbJ}{\mathbb{J}} \nc{\bbK}{\mathbb{K}} \nc{\bbL}{\mathbb{L}}
 \nc{\bbM}{\mathbb{M}} \nc{\bbN}{\mathbb{N}} \nc{\bbO}{\mathbb{O}}
 \nc{\bbP}{\mathbb{P}} \nc{\bbQ}{\mathbb{Q}} \nc{\bbR}{\mathbb{R}}
 \nc{\bbS}{\mathbb{S}} \nc{\bbT}{\mathbb{T}} \nc{\bbU}{\mathbb{U}}
 \nc{\bbV}{\mathbb{V}} \nc{\bbW}{\mathbb{W}} \nc{\bbX}{\mathbb{X}}
 \nc{\bbZ}{\mathbb{Z}}

 \nc{\bA}{{\bf A}} \nc{\bB}{{\bf B}} \nc{\bC}{{\bf C}}
 \nc{\bD}{{\bf D}} \nc{\bE}{{\bf E}} \nc{\bF}{{\bf F}}
 \nc{\bG}{{\bf G}} \nc{\bH}{{\bf H}} \nc{\bI}{{\bf I}}
 \nc{\bJ}{{\bf J}} \nc{\bK}{{\bf K}} \nc{\bL}{{\bf L}}
 \nc{\bM}{{\bf M}} \nc{\bN}{{\bf N}} \nc{\bO}{{\bf O}}
 \nc{\bP}{{\bf P}} \nc{\bQ}{{\bf Q}} \nc{\bR}{{\bf R}}
 \nc{\bS}{{\bf S}} \nc{\bT}{{\bf T}} \nc{\bU}{{\bf U}}
 \nc{\bV}{{\bf V}} \nc{\bW}{{\bf W}} \nc{\bX}{{\bf X}}
 \nc{\bZ}{{\bf Z}}

\nc{\cA}{{\cal A}} \nc{\cB}{{\cal B}} \nc{\cC}{{\cal C}}
\nc{\cD}{{\cal D}} \nc{\cE}{{\cal E}} \nc{\cF}{{\cal F}}
\nc{\cG}{{\cal G}} \nc{\cH}{{\cal H}} \nc{\cI}{{\cal I}}
\nc{\cJ}{{\cal J}} \nc{\cK}{{\cal K}} \nc{\cL}{{\cal L}}
\nc{\cM}{{\cal M}} \nc{\cN}{{\cal N}} \nc{\cO}{{\cal O}}
\nc{\cP}{{\cal P}} \nc{\cQ}{{\cal Q}} \nc{\cR}{{\cal R}}
\nc{\cS}{{\cal S}} \nc{\cT}{{\cal T}} \nc{\cU}{{\cal U}}
\nc{\cV}{{\cal V}} \nc{\cW}{{\cal W}} \nc{\cX}{{\cal X}}
\nc{\cZ}{{\cal Z}}

\nc{\hA}{{\hat{A}}} \nc{\hB}{{\hat{B}}} \nc{\hC}{{\hat{C}}}
\nc{\hD}{{\hat{D}}} \nc{\hE}{{\hat{E}}} \nc{\hF}{{\hat{F}}}
\nc{\hG}{{\hat{G}}} \nc{\hH}{{\hat{H}}} \nc{\hI}{{\hat{I}}}
\nc{\hJ}{{\hat{J}}} \nc{\hK}{{\hat{K}}} \nc{\hL}{{\hat{L}}}
\nc{\hM}{{\hat{M}}} \nc{\hN}{{\hat{N}}} \nc{\hO}{{\hat{O}}}
\nc{\hP}{{\hat{P}}} \nc{\hR}{{\hat{R}}} \nc{\hS}{{\hat{S}}}
\nc{\hT}{{\hat{T}}} \nc{\hU}{{\hat{U}}} \nc{\hV}{{\hat{V}}}
\nc{\hW}{{\hat{W}}} \nc{\hX}{{\hat{X}}} \nc{\hZ}{{\hat{Z}}}

\nc{\hn}{{\hat{n}}}

\nc{\as}{{\cal AS}}
	\nc{\app}{{\cal AP}}

\def\max{\mathop{\rm max}}

\def\rank{\mathop{\rm rank}}

\def\tr{\mathop{\rm Tr}}

\def\dg{\dagger}

\newcommand{\bra}[1]{\langle#1|}
\newcommand{\ket}[1]{|#1\rangle}
\newcommand{\proj}[1]{| #1\rangle\!\langle #1 |}
\newcommand{\ketbra}[2]{|#1\rangle\!\langle#2|}

\allowdisplaybreaks[4]

\hypersetup{colorlinks,linkcolor={blue},citecolor={blue},urlcolor={red}}

\begin{document}

\title{Entanglement distillation on symmetric two-qutrit entangled states of rank five}
\author[1]{Zihua Song}
\author[1*]{Lin Chen}
\author[1*]{Yongge Wang}

\affil[1]{LMIB(Beihang University), Ministry of Education,
and School of Mathematical Sciences, Beihang University, Beijing 100191, China}
\affil[*]{Corresponding authors: linchen@buaa.edu.cn (Lin Chen); wangyongge@buaa.edu.cn (Yongge Wang)}

\date{\today}

\maketitle


\begin{abstract}
Entanglement distillation is a key step in quantum information, both theoretically and practically. It has been proven that non-positive-partial transpose (NPT) entangled states of rank at most four is 1-distillable under local operation and classical communications. In this paper we investigate the distillation of a more complex family of NPT entangled states, namely a family of symmetric two-qutrit states $\r$ of rank five with given eigenvectors. We explicitly construct five families of such states by requiring four of the five eigenvalues to be the same. We respectively show that some of them are 1-distillable. It turns out that such states may be not 1-distillable for some interval of eigenvalues. We provide some conditions for eigenvalues that allow $\r$ to be 1-distillable or to be 1-undistillable.
\end{abstract}



\section{Introduction}

Entanglement has been used in various quantum-information applications \cite{bbc93} and fundamental of quantum theories \cite{werner89} in the past decades.
Constructing the theoretical tools for detection of entanglement is a key task. One of such tools is the so-called partial transpose map, namely \cite{peres1996}. It says that a non-entangled (i.e., separable) state has positive partial transpose (PPT)
\cite{hhh96}, and the converse also holds for two-qubit and qubit-qutrit systems \cite{Horodecki1997Reduction}. Apart from the entanglement detection, it has been shown that many quantum-information tasks require pure rather than mixed entanglement, where the latter is a more usual form of quantum correlation due to the unavoidable noise from nature. Hence, extracting pure entangled states (such as Bell states) from mixed entangled states, called entanglement distillation, is a basic task for quantum-information processing. The mixed states with extractable pure entanglement are said to be distillable. It has been shown all two-qubit entangled states and thus non-PPT (NPT) $2\times n$ entangled states are distillable. \cite{1997Inseparable}. Further, the states violating the reduction criterion are also distillable \cite{hhh1998}.
The quantitative estimation of distillable entanglement and secret key have been studied \cite{rains1999,rains2001,dw2005}.
Experimental entanglement distillation related to nonlocality was also proposed \cite{Kwiat2001Experimental}. Further, the restriction over the distillability of bipartite marginals of a tripartite state has been studied \cite{ch12ijmpb}.

On the other hand, all bipartite states can be converted into Werner states under local operation and classical communications (LOCC), thus it suffices to distill Werner states \cite{2000Evidence}. However, such states have full rank and its distillability has turned out to be hard to be determined. Actually, the distillability problem has been proposed as one of the five key theoretical problems proposed recently \cite{2002.03233}. The main difficulty of distillability problem lies in the fact that many copies of target states are required in the definition of entanglement distillation \cite{watrous04,vd06,Chen2018Generalized}
, which makes the problem mathematically complex.
As a result, researchers have considered the distillability of states in terms of their matrix rank. For example, rank-two, three and four NPT entangled states have been shown to be distillable
\cite{1999Rank,Lin2008Rank,cd16pra}. The distillability of states with low rank in high dimensions have also been studied \cite{cd11jpa}.

In this paper, we go a step further in the above line, that is, we consider to distill NPT entangled states $\r$ of rank five. 
We begin by introducing the existing facts on 1-distillability in Lemmas \ref{le:mxnNPT=distill} and \ref{le:3x3NPT=oneUndis}. They respectively work for high-dimensional bipartite and two-qutrit systems. Lemma \ref{le:symmU=V} presents a method for the simplification of specific quantum states via unitary matrices. We also provide facts on linear algebra including partial transpose in Lemmas \ref{le:hermi=posi}-\ref{le:(m-1)(n-1)+1=prodvector}. Next, we show in Lemma \ref{le:0001span} that
if $\r$ is 1-undistillable, then the range of $\r$ has no the two-dimensional subspace spanned by $\ket{00}$ and $\ket{01}$. It implies the result in Theorem \ref{th:rho=3x3symRank5}. That is, if $\r$ has the range in the symmetric subspace, then $\r$ is 1-distillable when its kernel has a product vector. We explicitly list three cases for the kernel. By excluding the above distillable states, we arrive at the hardest part of this paper. We explicitly construct a family $\r$ diagonal in the five orthonormal pure states in equation \eqref{eq:kete1} with positive eigenvalues $\l_1,...,\l_5$. The main results are $\r$ must be 1-distillable when certain eigenvalues are the same and we also provide conditions for eigenvalues that allow $\r$ to be 1-distillable or to be 1-undistillable. In Lemma \ref{le:l12=x}, we show that $\r$ is 1-distillable if for any $i \in \{1,2,3,4\}$, $\lambda_j = \frac{1 - \lambda_i}{4}$ for all $j \neq i$. When $\l_1=\l_2=\l_3=\l_4=\frac{1-\l_5}{4}$, the situation becomes more complex. We show in Proposition \ref{pr:l5=x} that $\r$ is 1-distillable for the interval $\l_5\in[0,\frac{24\sqrt{2}-33}{7})\cup(\frac{33-12\sqrt{6}}{25},1]$. The distillability of $\r$ with $\l_5$ not in the interval is unknown yet. We show in Example \ref{ex:l5=x} that there indeed exists at least one 1-undistillable $\r$, which may offer inspiration for future research.

The rest of this section is organized as follows. In Sec. \ref{sec:pre} we introduce the primary facts and knowledge used in this paper. In Sec. \ref{sec:res} we introduce the main result of this paper. We finally conclude in Sec. \ref{sec:con}.

\section{Preliminaries}
\label{sec:pre}

In this section, we introduce the main technique and facts used in this paper. Let \(\mathcal{H}=\mathcal{H}_A \otimes \mathcal{H}_B\) be the bipartite Hilbert space with \(\mathrm{dim} \mathcal{H}_A=M\) and \(\mathrm{dim} \mathcal{H}_B=N\). We study bipartite quantum states \(\rho\) on \(\mathcal{H}\). We denote the range and kernel of a linear map \(\rho\) with \(\mathcal{R}(\rho)\) and \(\ker \rho\), respectively. Unless stated otherwise, the states will not be normalized. We denote orthonormal bases of \(\mathcal{H}_A\) and \(\mathcal{H}_B\) with \(\{|i\rangle_A:i=0,1,\cdots,M-1\}\) and \(\{|j\rangle_B:j=0,1,\cdots,N-1\}\), respectively. The partial transpose of \(\rho\) with respect to the system \(A\) is defined as \(\rho^{\Gamma} := \Sigma^{M-1}_{i,j=0} (\ketbra{i}{j}\otimes I_n)\r(\ketbra{i}{j}\otimes I_n)\). We say that \(\rho\) is partial-positive-transpose (PPT) if \(\rho^{\Gamma}\ge0\). Otherwise, \(\rho\) is NPT, that is, the Hermitian matrix \(\rho^\G\) has at least one negative eigenvalue. The NPT states are always entangled, because non-entangled (i.e., separable) states are the convex sum of product states.

The distillability problem requires many-copy states from a composite system. Let \(\rho_{A_iB_i}\) be an \(M_i \times N_i\) state of rank \(r_i\) acting on the Hilbert space \(\mathcal{H}_{A_i} \otimes \mathcal{H}_{B_i}\), \(i=1,2\). Suppose \(\rho\) of systems \(A_1,A_2\) and \(B_1,B_2\) is a state acting on the Hilbert space \(\mathcal{H}_{A_1} \otimes \mathcal{H}_{B_1} \otimes \mathcal{H}_{A_2} \otimes \mathcal{H}_{B_2}\), such that \(\tr_{A_1B_1} \rho = \rho_{A_2B_2}\) and \(\tr_{A_2B_2} \rho = \rho_{A_1B_1}\). By switching the two middle factors, we can consider \(\rho\) a composite bipartite state acting on the Hilbert space \(\mathcal{H}_A \otimes \mathcal{H}_B\), where \(\mathcal{H}_A=\mathcal{H}_{A_1} \otimes \mathcal{H}_{A_2}\) and \(\mathcal{H}_B=\mathcal{H}_{B_1} \otimes \mathcal{H}_{B_2}\). In that case we shall write \(\rho=\rho_{A_1A_2:B_1B_2}\). So \(\rho\) is an \(M_1M_2 \times N_1N_2\) state of rank not larger than \(r_1r_2\). In particular for the tensor product \(\rho=\rho_{A_1B_1} \otimes \rho_{A_2B_2}\), it is easy to see that \(\rho\) is an \(M_1M_2 \times N_1N_2\) state of rank \(r_1r_2\). 
The above definition can be easily generalized to the tensor product of \(N\) states \(\rho_{A_iB_i}\), \(i=1,\cdots,N\). They form a bipartite state on the Hilbert space \(\mathcal{H}_{A_1,\cdots,A_N} \otimes \mathcal{H}_{B_1,\cdots,B_N}\). It is written as \(\mathcal{H}^{\otimes n}\) when \(\mathcal{H}_{A_i} \otimes \mathcal{H}_{B_i}=\mathcal{H}\). Now we can define the distillability of entangled states.

\begin{definition}
    A bipartite state \(\rho\) is n-distillable under LOCC if there exists a Schmidt-rank-two state \(|\ps \rangle \in \mathcal{H}^{\otimes n}\) such that \(\langle \ps |(\rho^{\otimes n})^{\G}|\ps\rangle <0\). Here, the state $\ket{\ps}$ can be written as the superposition of two pure product states. Equivalently, \(\rho\) is n-distillable under LOCC if there exists a rank-two projection operator $P$ on subsystem $A^n$ such that the matrix $(P\otimes I_{B^n})(\rho^{\otimes n})^{\G}(P^\dg\otimes I_{B^n})$ has at least one negative eigenvalue. If a finite $n$ exists then we say that $\r$ is distillable. If no such $n$ exists, $\r$ is called undistillable.
\qed
\end{definition}
If an entangled state \(\rho\) is not distillable, then we say that $\r$ is bound entangled. For example, PPT entangled states are bound entangled states. The long-standing distillability problem asks whether a bound entangled state can be NPT. To simplify the  problem, we convert one entangled state into another under stochastic LOCC (SLOCC) or equivalently, product general linear group (PGL). Formally, two bipartite states \(\rho\) and \(\s\) are equivalent under SLOCC if there exists an invertible local operator (ILO) \(A \otimes B\) such that \(\rho=(A^* \otimes B^*)\s(A \otimes B)\). 
From the definition of distillable entangled states, we see that SLOCC equivalent states are distillable at the same time. 

Next we introduce some facts on states with separability and distillable entanglement. 
We shall refer to the positive (resp. zero, negative) subspace of an Hermitian matrix $H$ as the subspace spanned by the eigenvectors of positive (resp. zero, negative) eigenvalues of $H$.
\begin{lemma}
\label{le:mxnNPT=distill}
Let $2\le m \le n$. If an $m\times n$ NPT state $\r$ is one-undistillable, then

(i) $m>2$;

(ii) $\rank \r>4$;

(iii) $\rank \r> \max\{\rank\r_A,\rank\r_B\}$;

(iv) the negative subspace of $\r^\G$ contains no vector of Schmidt rank one or two.
\qed
\end{lemma}

The following fact is from \cite{cd16pra}.

\begin{lemma}
\label{le:3x3NPT=oneUndis}
If a two-qutrit NPT state $\r$ is one-undistillable then 

(i) $\ker \rho$ has no product vector, thus $\rank\r>4$;

(ii) $\r^\G$ has exactly one negative eigenvalue and eight positive eigenvalues. As a corollary, $\rank\r\ge5$.
\qed
\end{lemma}

\begin{lemma}
\label{le:symmU=V}
Let $\ket{a}\in\mathrm{span}\{\ket{jj},\ket{ik}+\ket{ki}\}\subset\bbC^d\otimes\bbC^d$ be a non-normalized symmetric state. Then we can find 
a unitary $U=V\in\cU^d$ that $\ket{a}=(U\otimes V)\S_{j=0}^{d-1}\sqrt{s_j}\ket{j,j}$.
\end{lemma}

To conclude this section, we review some essential results from linear algebra.

\begin{lemma}
\label{le:hermi=posi}
For a given $n$-order Hermitian matrix $A$, we have

(i) $A$ is positive semidefinite if and only if all principal minors of $A$ are nonnegative;

(ii) If the first $n-1$ leading principal minors (respectively, the last $n-1$ trailing principal minors) of $A$ are positive and $\det A\ge0$, then $A$ is positive semidefinite.
\qed
\end{lemma}

\begin{lemma}
\label{le:(m-1)(n-1)+1=prodvector}
The $((m-1)(n-1)+1)$-dimensional bipartite subspace in $\bbC^m\otimes\bbC^n$ has at least one product vector.
\qed
\end{lemma}

\section{Result}
\label{sec:res}

In this section, we investigate the 1-distillability of rank-five NPT states $\r$.  Lemma \ref{le:0001span} establishes that if $\rho$ is 1-undistillable, $\mathcal{R}(\r)$ cannot contain the two-dimensional subspace spanned by $\ket{00}$ and $\ket{01}$. This leads to Theorem \ref{th:rho=3x3symRank5}, which proves that a symmetric rank-five two-qutrit NPT state $\rho$ is 1-distillable if $\ker \rho$ contains a product vector, through analysis of three explicit cases.

We then construct a family of such states $\rho$, diagonal in five orthonormal states with positive eigenvalues $\l_1,...,\l_5$, via equation \eqref{eq:rho} and \eqref{eq:kete1}. The main results of this section are as follows: Lemma \ref{le:0001span} establishes that $\r$ is 1-distillable if for any $i \in \{1,2,3,4\}$, $\lambda_j = \frac{1 - \lambda_i}{4}$ for all $j \neq i$. For the case where $\l_1=\l_2=\l_3=\l_4=\frac{1-\l_5}{4}$, Proposition \ref{pr:l5=x} confirms 1-distillability for $\lambda_5 \in [0, \frac{24\sqrt{2} - 33}{7}) \cup (\frac{33 - 12\sqrt{6}}{25}, 1]$, while the status for $\lambda_5 \in [\frac{24\sqrt{2} - 33}{7}, \frac{33 - 12\sqrt{6}}{25}]$ remains open. Example \ref{ex:l5=x} provides a special case within this unresolved interval, showing that when $x = \frac{1}{7}$, $\rho$ is 1-undistillable under the assumptions of Proposition \ref{pr:l5=x}.

\begin{lemma}
\label{le:0001span}
Let $\rho$ be a two-qutrit NPT state of rank five that is 1-undistillable and satisfies the conditions in Lemmas \ref{le:mxnNPT=distill} and \ref{le:3x3NPT=oneUndis}. Then $\mathcal{R}(\rho)$ does not contain the two-dimensional subspace spanned by $\ket{00}$ and $\ket{01}$.
\end{lemma}
\begin{proof}
In contrast, assume $\mathcal{R}(\rho)$ contains the two-dimensional subspace $S$ spanned by $\ket{00}$ and $\ket{01}$. Let ${\ket{00}, \ket{01}, \ket{a}, \ket{b}, \ket{c}}$ be a basis for $\mathcal{R}(\rho)$.
We can express the vectors $\ket{a}, \ket{b}, \ket{c}$ in the following form
\begin{align*} 
    \ket{a}&=\ket{0x_1}+\ket{1x_2}+\ket{2x_3},\\
    \ket{b}&=\ket{0y_1}+\ket{1y_2}+\ket{2y_3},\\
    \ket{c}&=\ket{0z_1}+\ket{1z_2}+\ket{2z_3},
\end{align*}
for some vectors $\ket{x_i}, \ket{y_i}, \ket{z_i} \in \mathbb{C}^3$ ($i = 1,2,3$).

Now we consider the set of vectors
\[\{\ket{1x_2}+\ket{2x_3},\ket{1y_2}+\ket{2y_3},\ket{1z_2}+\ket{2z_3}\}\subset \bbC^2\otimes \bbC^3\]
The dimension of its orthogonal complement $V$ within $\mathbb{C}^2 \otimes \mathbb{C}^3$ is at least $3$. By Lemma \ref{le:(m-1)(n-1)+1=prodvector}, the subspace $V$ must contain at least one product vector. 

Let $\ket{v} \in V$ be such a product vector. We can write it as 
\[\ket{v} = (m\ket{1} + n\ket{2}) \otimes \ket{w},\] 
where $m,n\in\bbC$ and $\ket{w} \in \mathbb{C}^3$.

By construction, $\ket{v}$ is orthogonal to each of the vectors $\ket{1x_2} + \ket{2x_3}$, $\ket{1y_2} + \ket{2y_3}$, and $\ket{1z_2} + \ket{2z_3}$. Moreover, $\ket{v}$ is clearly orthogonal to $\ket{0} \otimes \mathbb{C}^3$. It follows that $\ket{v}$ is orthogonal to the entire basis $\{{\ket{00}, \ket{01}, \ket{a}, \ket{b}, \ket{c}}\}$ of $\mathcal{R}(\rho)$. Therefore, $\ket{v} \in \ker \rho$.

However, this contradicts Lemma \ref{le:3x3NPT=oneUndis}, which asserts that the kernel of a 1-undistillable two-qutrit NPT state of rank five cannot contain any product vector. It completes the proof.
\end{proof}

\begin{theorem}
\label{th:rho=3x3symRank5}
Let $\r$ be a two-qutrit symmetric NPT state of rank five. Then $\r$ is 1-distillable when $\ker \rho$ has a product vector. This condition is satisfied in any of the following three cases

(i) $\mathcal{R}(\rho)$ is spanned by $\ket{00},\ket{11},\ket{01}+\ket{10},\ket{02}+\ket{20}$ and $\ket{12}+\ket{21}$;  

(ii) $\mathcal{R}(\rho)$ is spanned by $\ket{00},\ket{11},\ket{22},\ket{02}+\ket{20}$, and $\ket{12}+\ket{21}$;

(iii) $\ker \rho$ has a symmetric pure state of Schmidt rank at most two.

Conversely, $\ker \rho$ contains no product vector if it contains a symmetric pure state of Schmidt rank three.
\end{theorem}
\begin{proof}
\textbf{(i)} $\ker \rho$ can be verified to contain the product vector $\ket{22}$. Then $\rho$ is 1-distillable by Lemma \ref{le:3x3NPT=oneUndis}(i).

\textbf{(ii)} $\ker \rho$ can be verified to contain the product vector $\ket{01}$. Then $\rho$ is 1-distillable by Lemma \ref{le:3x3NPT=oneUndis}(i).

It should be noted that the NPT condition may not be assumed separately for (i) since the state $\rho$ in (i) is NPT by construction. In contrast, the NPT assumption is essential for (ii), since symmetric PPT states satisfying (ii) do exist.

\textbf{(iii)} Let $\ket{a}$ be such a pure state. By Lemma \ref{le:symmU=V}, there exists a product unitary operator $U_1 \otimes V_1$ such that 
\[(U_1 \otimes V_1)\ket{a} = \cos\theta \ket{00} + \sin\theta \ket{11},\] 
and the kernel of $(U_1 \otimes V_1) \rho (U_1 \otimes V_1)^\dagger$ is spanned by the antisymmetric subspace and $(U_1 \otimes V_1)\ket{a}$. The kernel can be verified to contain a product vector, and thus Lemma \ref{le:3x3NPT=oneUndis}(i) applies.

To prove the converse statement, suppose $\ker \rho$ contains a symmetric pure state $\ket{a}$ of Schmidt rank three, implying all Schmidt coefficients $s_j > 0$. By Lemma \ref{le:symmU=V}, we can transform $\ket{a}$ into a diagonal form
\[\ket{a}=(U_2\otimes V_2)\S_{j=0}^{2}\sqrt{s_j}\ket{j,j}.\]
Application of $U_2^{-1} \otimes V_2^{-1}$ to any vector $\ket{v} \in \ker \rho$ yields
\begin{eqnarray}
\label{eq:ketv}
\ket{v}=\a(\ket{01}-\ket{10})+\b(\ket{12}-\ket{21})+\g(\ket{02}-\ket{20})+\S_{j=0}^{2}\sqrt{s_j}\ket{j,j}.    
\end{eqnarray}
Let $\ket{v}=\S_{i,j}a_{ij}\ket{i,j}$. The coefficient matrix of $\ket{v}$ is
\[A=(a_{ij})=\left(
\begin{matrix}
    \sqrt{s_0}& \a& \g\\
    -\a& \sqrt{s_1}& \b\\
    -\g& -\b& \sqrt{s_2}
\end{matrix}
\right).\]
If $\ket{v}$ is a product vector, then $A$ must have rank one. So all $2 \times 2$ minors of $A$ vanish. This yields
\begin{align*}
\alpha^2 &= -\sqrt{s_0 s_1}, \\
\beta^2 &= -\sqrt{s_1 s_2}, \\
\gamma^2 &= -\sqrt{s_0 s_2}.
\end{align*}
Moreover, the following condition is obtained
\begin{align}
\label{eq:detleft}
    \det\left(
\begin{matrix}
    a_{00}& a_{01}\\
    a_{20}& a_{21}
\end{matrix}
\right)=-\b\sqrt{s_0}+\a\g=0.
\end{align}
Rearranging equation \eqref{eq:detleft} and squaring both sides gives
\[\a^2\g^2=\b^2s_0,\]
which implies
\[s_0\sqrt{s_1s_2}=0.\]
This contradicts the assumption that all $s_j > 0$. Therefore, no product vector $\ket{v}$ in \eqref{eq:ketv} exists in $\ker \rho$.
\end{proof}

To extend Theorem \ref{th:rho=3x3symRank5}, let $\rho$ be a two-qutrit symmetric NPT state of rank five that does not satisfy the preconditions of Lemmas \ref{le:mxnNPT=distill}, \ref{le:3x3NPT=oneUndis} or Theorem \ref{th:rho=3x3symRank5}. The question of whether such $\rho$ is 1-distillable remains a challenge. As the main focus of this paper, we consider states of the form
\begin{equation}
\label{eq:rho}
    \r=\S_{j=1}^5\lambda_j\ket{e_j}\bra{e_j},
\end{equation}
where the eigenvectors 
\begin{equation}
\label{eq:kete1}
    \begin{cases}
        \ket{e_1}&=\frac{\ket{01}+\ket{10}}{\sqrt{2}},\\
        \ket{e_2}&=\frac{\ket{12}+\ket{21}}{\sqrt{2}},\\
        \ket{e_3}&=\frac{\ket{02}+\ket{20}}{\sqrt{2}},\\
        \ket{e_4}&=\frac{\ket{00}-\ket{11}}{\sqrt{2}},\\
        \ket{e_5}&=\frac{\ket{00}+\ket{11}-2\ket{22}}{\sqrt{6}}.
    \end{cases}
\end{equation}
A natural approach for simplifying our problem is to set one eigenvalue $\lambda_j$ as a free variable and assign the remaining four eigenvalues as equal, i.e., $\lambda_i = (1 - \lambda_j)/4$ for $i \neq j$. This leads to the following five cases,
\begin{align*}
&(i) &\l_1=x,\ \l_2=\l_3=\l_4=\l_5=\frac{1-x}{4};\\
&(ii) &\l_4=x,\ \l_1=\l_2=\l_3=\l_5=\frac{1-x}{4};\\
&(iii) &\l_2=x,\ \l_1=\l_3=\l_4=\l_5=\frac{1-x}{4};\\
&(iv) &\l_3=x,\ \l_1=\l_2=\l_4=\l_5=\frac{1-x}{4};\\
&(v) &\l_5=x,\ \l_1=\l_2=\l_3=\l_4=\frac{1-x}{4}.
\end{align*}

We assert that cases (i) and (ii) are equivalent under local unitary transformations, and hence the states in both cases are 1-distillable. This equivalence follows from the observations
\begin{equation}
\begin{cases}
    K^{\otimes2}\ket{e_1}=\ket{e_4},\\
    K^{\otimes2}\ket{e_4}=\ket{e_1},\\
    K^{\otimes2}\ket{e_5}=\ket{e_5},\\
    K^{\otimes2}(\proj{e_2}+\proj{e_3})K^{\otimes2}=(\proj{e_2}+\proj{e_3}),
\end{cases}
\end{equation}
where $K$ is the $3 \times 3$ orthogonal matrix given by
\begin{eqnarray}
K={1\over\sqrt2}
\bma
1 & 1 \\
1 & -1 \\
\ema
\oplus I_1.
\end{eqnarray}
Similarly, cases (iii) and (iv) are also equivalent. In the following lemmas, we analyze cases (i), (iii), and (v) in detail. The first two cases turn out to yield distillable states, while the distillability of case (v) remains partially open.

\begin{lemma}
\label{le:l12=x}
Let $\rho$ be a two-qutrit symmetric NPT state of rank five, with an orthonormal basis $\{\ket{e_1}, \ldots, \ket{e_5}\}$ for $\mathcal{R}(\rho)$ given explicitly by
\begin{align}   
\ket{e_1} &= \frac{\ket{01} + \ket{10}}{\sqrt{2}}, \quad
\ket{e_2} = \frac{\ket{12} + \ket{21}}{\sqrt{2}}, \quad
\ket{e_3} = \frac{\ket{02} + \ket{20}}{\sqrt{2}}, \notag \\
\ket{e_4} &= \frac{\ket{00} - \ket{11}}{\sqrt{2}}, \quad
\ket{e_5} = \frac{\ket{00} + \ket{11} - 2\ket{22}}{\sqrt{6}}.
\end{align}
Suppose $\rho$ has the spectral decomposition $\rho = \sum_{j=1}^5 \lambda_j \proj{e_j}$, where there exists $i \in \{1, 2\}$ such that $\lambda_i = x\in(0,1)$ and $\lambda_j = \frac{1 - x}{4}$ for all $j \in\{1,2,3,4,5\}\backslash\{i\}$. Then $\rho$ is 1-distillable.
\end{lemma}

\begin{proof}
The proof follows a similar approach for both $i=1$ and $i=2$. We only focus on the case $i=1$. With this case, we begin by computing the partial transpose $\rho^\Gamma$ and its eigenvalues. The condition that $\rho^\Gamma$ has at least one negative eigenvalue, which is equivalent to $\rho$ being NPT, is found to be
\begin{equation*}
    x\in(0,\frac{1}{7})\cup(\frac{1}{4},1).
\end{equation*}
For such values of $x$, we construct the operator
\[A=\ket{0}\bra{0}+\ket{1}(\bra{1}+y\bra{2})\] with $y \in \mathbb{C}$, and consider the projected state
\[\s=(A\otimes I)\r(A^\dg \otimes I).\] 
By Lemma \ref{le:hermi=posi}, there exists a choice of $y$ such that $\sigma$ is NPT for each $x \in (0, \tfrac{1}{7}) \cup (\tfrac{1}{4}, 1)$. This implies that $\rho$ is 1-distillable for all such $x$.
\end{proof}

Lemma \ref{le:l12=x} addresses the case where the free eigenvalue parameter $x$ is assigned to one of the vectors $\ket{e_1}$, $\ket{e_2}$, $\ket{e_3}$ or $\ket{e_4}$. Then we examine the alternative configuration where the distinguished eigenvalue $x$ is associated with the vector $\ket{e_5}$. It will be more special and complex than Lemma \ref{le:l12=x}.

\begin{proposition}
\label{pr:l5=x}
Let $\rho$ be a two-qutrit symmetric NPT state of rank five, with an orthonormal basis $\{\ket{e_1}, \ldots, \ket{e_5}\}$ for $\mathcal{R}(\rho)$ given explicitly by
\begin{align}   
\ket{e_1} &= \frac{\ket{01} + \ket{10}}{\sqrt{2}}, \quad
\ket{e_2} = \frac{\ket{12} + \ket{21}}{\sqrt{2}}, \quad
\ket{e_3} = \frac{\ket{02} + \ket{20}}{\sqrt{2}}, \notag \\
\ket{e_4} &= \frac{\ket{00} - \ket{11}}{\sqrt{2}}, \quad
\ket{e_5} = \frac{\ket{00} + \ket{11} - 2\ket{22}}{\sqrt{6}}.
\end{align}
Suppose $\rho$ has the spectral decomposition $\rho = \sum_{j=1}^5 \lambda_j \proj{e_j}$, where $\lambda_5 = x$ and $\lambda_1 = \lambda_2 = \lambda_3 = \lambda_4 = \frac{1 - x}{4}$ for some $x \in (0,\frac{24\sqrt{2}-33}{7})\cup(\frac{33-12\sqrt{6}}{25},1)$. Then $\rho$ is 1-distillable.
\end{proposition}
\begin{proof}
We begin with the expression for $\rho$
\begin{equation*}
    \r=x\ket{e_5}\bra{e_5}+\frac{1-x}{4}\S_{j=1}^4\ket{e_j}\bra{e_j}.
\end{equation*}
By calculation, the condition for $\rho^\Gamma$ to have at least one negative eigenvalue (i.e., for $\rho$ to be NPT) is
\[x \in (0, \frac{33-12\sqrt{6}}{25}) \cup (\frac{3}{11}, 1),\] 
and the condition for at least two negative eigenvalues (i.e., for $\rho$ to be 1-distillable) is
\[x \in (\frac{3}{11}, 1).\]
We now consider the projection operator
\[A = \ket{0}\bra{0} + \ket{1}(\bra{1} + y\bra{2})\]
with $y \in \mathbb{C}$, and define the transformed state
\[\s=(A\otimes I)\r(A^\dg \otimes I).\]
All first-order principal minors of $\s^\G$ are nonnegative. For $\s^\G$ to be positive semidefinite, all second-order principal minors must be positive
\begin{equation}
\label{eq:l5=x}
    \begin{cases}
        \frac{x(1-x)}{48}+\frac{x^2|y|^2}{9}+\frac{(1-x)^2}{64}+\frac{x(1-x)|y|^2}{12}-\frac{(1-x)^2|y|^2}{64}\ge0,\\
        \frac{(1-x)^2(1+|y|^2)}{64}-\frac{x^2}{36}+\frac{x(1-x)}{24}-\frac{(1-x)^2}{64}\ge0,\\
        \frac{(1-x)^2(1+|y|^2)}{64}-\frac{x^2|y|^2}{9}\ge0.
    \end{cases}
\end{equation}
For these inequalities to hold for all $y$, the following conditions must be met
\[\frac{24\sqrt{2}-33}{7}\le x \le \frac{3}{11}.\]
Thus, for $x \in [\frac{24\sqrt{2}-33}{7}, \frac{33-12\sqrt{6}}{25}] \approx [0.134, 0.144]$, $\rho$ is NPT but its 1-distillability cannot be established using the projection operator $A \otimes I$.
\end{proof}

In spite of above results with distillable states, we provide an example of 1-undistillable state below. We also conjecture that the state is undistillable with any copies. 

\begin{example}
\label{ex:l5=x}
The given $\r$ in Proposition \ref{pr:l5=x} is $1$-undistillable if $x=\frac{1}{7}\in(\frac{24\sqrt{2}-33}{7},\frac{33-12\sqrt{6}}{25})$. We define the following transformation matrices
\[K=\dfrac{1}{\sqrt{2}}
\begin{pmatrix}
    1 & i \\
    1 & -i \\
\end{pmatrix},\ \ \ 
\overline{K}=\dfrac{1}{\sqrt{2}}
\begin{pmatrix}
    1 & -i \\
    1 & i \\
\end{pmatrix}.\]
Then we transform $\rho$ into a simpler form 
\begin{equation}
    \widetilde{\rho}:=(K\otimes \overline{K})\r(K\otimes \overline{K})^\dg.
\end{equation}
Based on PGL equivalence and the constraints that ensure the inertia of $\rho^\Gamma$ is $(1,0,8)$, the projection $P$ on subsystem $A$ can be restricted to one of the following two forms
\begin{eqnarray}
P_1=\bma 
1 & a & 0 \\
0 & 0 & 1 \\
0 & 0 & 0
\ema,    
\ \ \ 
P_2=\bma 
1 & 0 & b \\
0 & 1 & c \\
0 & 0 & 0
\ema,   
\end{eqnarray}
where $a, b, c \in \mathbb{C}$ are parameters to be chosen such that the projected matrix is not positive semidefinite. We will sequentially prove that $\widetilde{\rho}$ is 1-undistillable in both cases.
\end{example}
\begin{proof}
    Continuing from Example \ref{ex:l5=x}, we substitute $x = \frac{1}{7}$ into $\widetilde{\rho}$ and define the resulting matrices
\begin{eqnarray}
&\a_i:=(P_i\otimes I_3) (\widetilde{\rho}_{x=\frac{1}{7}})^\Gamma (P_i^\dg \otimes I_3), \ i=1,2.
\end{eqnarray}
By Lemma \ref{le:hermi=posi}, $\alpha_1$ is positive semidefinite for all $a$. For $\alpha_2$, we apply Lemma \ref{le:hermi=posi}(ii) and examine its leading principal minors. The first three minors are readily verified to be positive. The fourth minor of $\alpha_2$, denoted as $\alpha_{2}^{4}$, is
\begin{align}
    \nonumber
    \a_{2}^{4}=&\frac{1}{5531904}[737 + 268b^2 + 648|b|^6 + 715|c|^2 \\
    \nonumber
    &+ |b|^4(1184 + 63|c|^2) 
    + \overline{b}(b(1427 + 324b^2) \\
    &+ 4\overline{b}(67 + 81|b|^2) + 778b|c|^2)].
\end{align}
$\alpha_{2}^{4}$ is positive for all complex $b, c$, and in particular never vanishes. 

The fifth minor of $\alpha_2$, denoted as $\alpha_{2}^{5}$, is
\begin{align}
\nonumber
\a_{2}^{5}=&\frac{1}{464679936}(536 (5 + b^2) \\
\nonumber
&+ |b|^2 (8533 + 648b^2 - 504c^2)+ 9 \left|c\right|^4 (715 + 63|b|^2)\\
\nonumber
&+ \overline{b}^2 (536 + 6868b^2 + 4095c^2 + 81|b|^2 (8 + 7b^2 - 7c^2))\\
\nonumber
&+ |c|^2 (9233 + 2412b^2 + 4788 \left|b\right|^4 + 7388|b|^2 + 2412\overline{b}^2)\\
&- 18b (-65b + 9b \left|b\right|^2 + 8\overline{b}) \overline{c}^2).
\end{align}
By calculation, $\alpha_{2}^{5}$ is positive for all real $b$ and $c$. 

Finally, the determinant of $\alpha_2$ is
\begin{align}
    \nonumber
    \det\a_{2}=&\frac{1}{39033114624}(585b(8 + 7b^2 - 7c^2)\overline{b}^3 \\
    \nonumber
    & + \overline{b}^2(4824 + 50436b^2 + 30647c^2 \\
    \nonumber
    &- 65|c|^2(-212 - 595b^2 + 63c^2))\\
    \nonumber
    &+ |b|^2(71037 + 4680b^2 + 13780c^2 \\
    \nonumber
    &+ 38675|c|^4 + 56293|c|^2 - 1170(-14 + b^2)\overline{c}^2)\\
    \nonumber
    &+ 9(536(5 + b^2 + c^2) + \overline{c}(c(7893 + 1820b^2 + 520c^2) \\
    &+ \overline{c}(536 + 1058b^2 + 5604c^2 + 65|c|^2(8 - 2b^2 + 7c^2))))).
\end{align}
As with the fifth minor, $\det\a_{2}$ is positive for all real $b, c$. 

To fully analyze the value distribution of $\alpha_{2}^{5}$ and $\det\a_{2}$ when $b,c\in\bbC$, their graphical representations can be used for verification.

Figures \ref{fig:a25} and \ref{fig:a26} plot $1075648\a_{2}^{5}$ and $7529536\det\a_{2}$ as functions of $b$ for fixed values of $c$, with the z-axis starting at $1$. The displayed range is restricted to function values between $1$ and $40000$. Neither figure shows any discontinuities, indicating that $\a_{2}^{5}>1>0$ and $\det\a_{2}>1>0$ always hold. 
Furthermore, Figures \ref{fig:a255} and \ref{fig:a265} suggest that the minimum values of both $1075648\a_{2}^{5}$ and $7529536\det\a_{2}$ lie within the interval $[1,10]$.
\end{proof}

\section{Conclusions}
\label{sec:con}

We have investigated the distillation of a family of symmetric NPT entangled two-qutrit states $\r$ of rank five with given eigenvectors. We explicitly constructed five families of such states by requiring four of the five eigenvalues to be the same, and showed their 1-distillability. We also constructed some symmetric $\r$ which may be not 1-distillable. The next target from this paper is to establish the explicit interval of eigenvalues of $\r$ which is not 1-distillable. Another interesting problem is to construct more two-qutrit symmetric states which are 1-distillable.

\section*{Acknowledgements} 

ZHS and LC were supported by the NNSF of China (Grant No. 12471427).

\begin{figure}[H]
    \centering
    \begin{subfigure}{0.9\textwidth}
        \includegraphics[width=\textwidth]{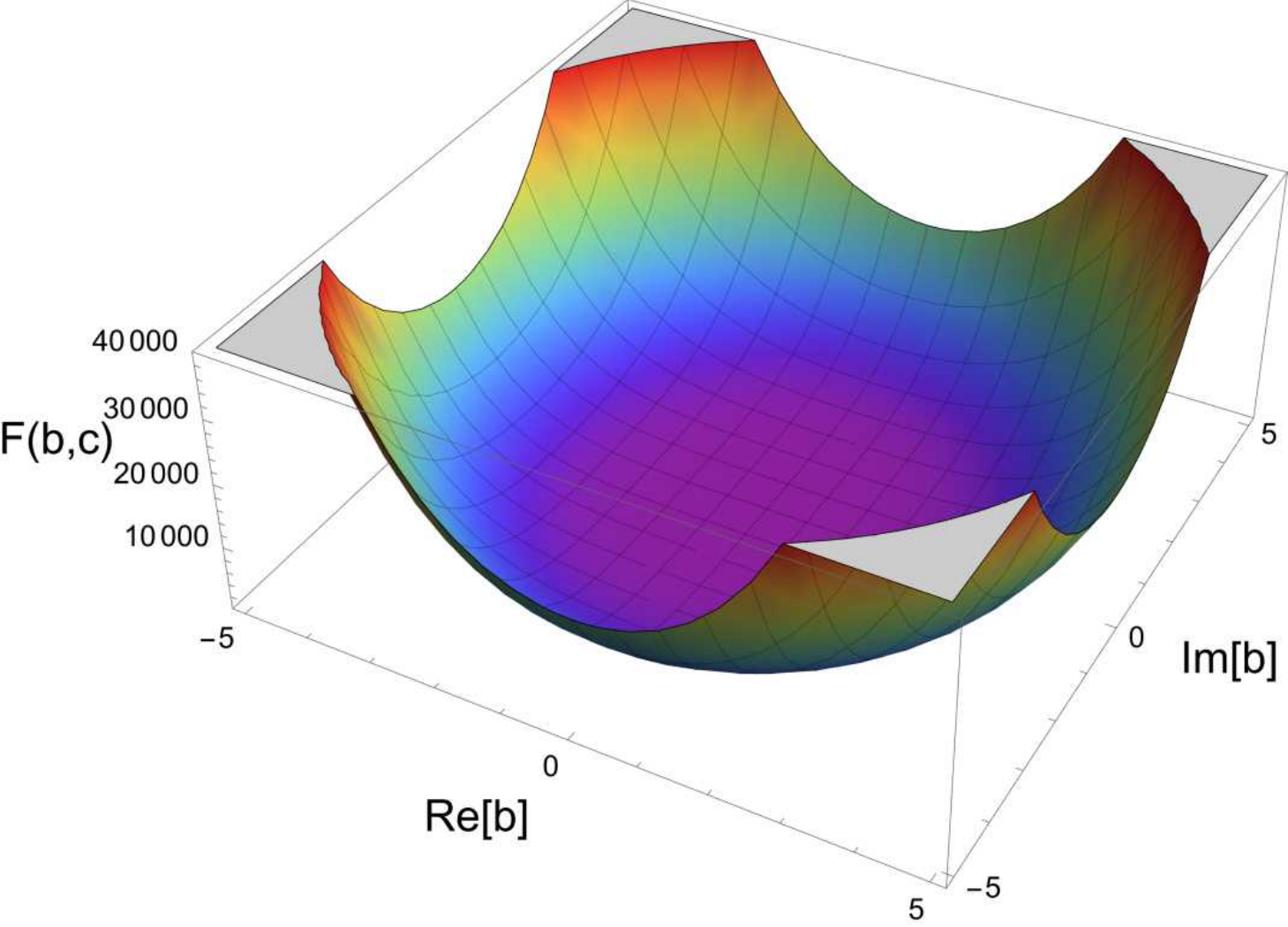}
        \caption{$c=0$}
        \label{fig:a250}
    \end{subfigure}

    \vspace{0.5cm}
    \begin{subfigure}{0.45\textwidth}
        \includegraphics[width=\textwidth]{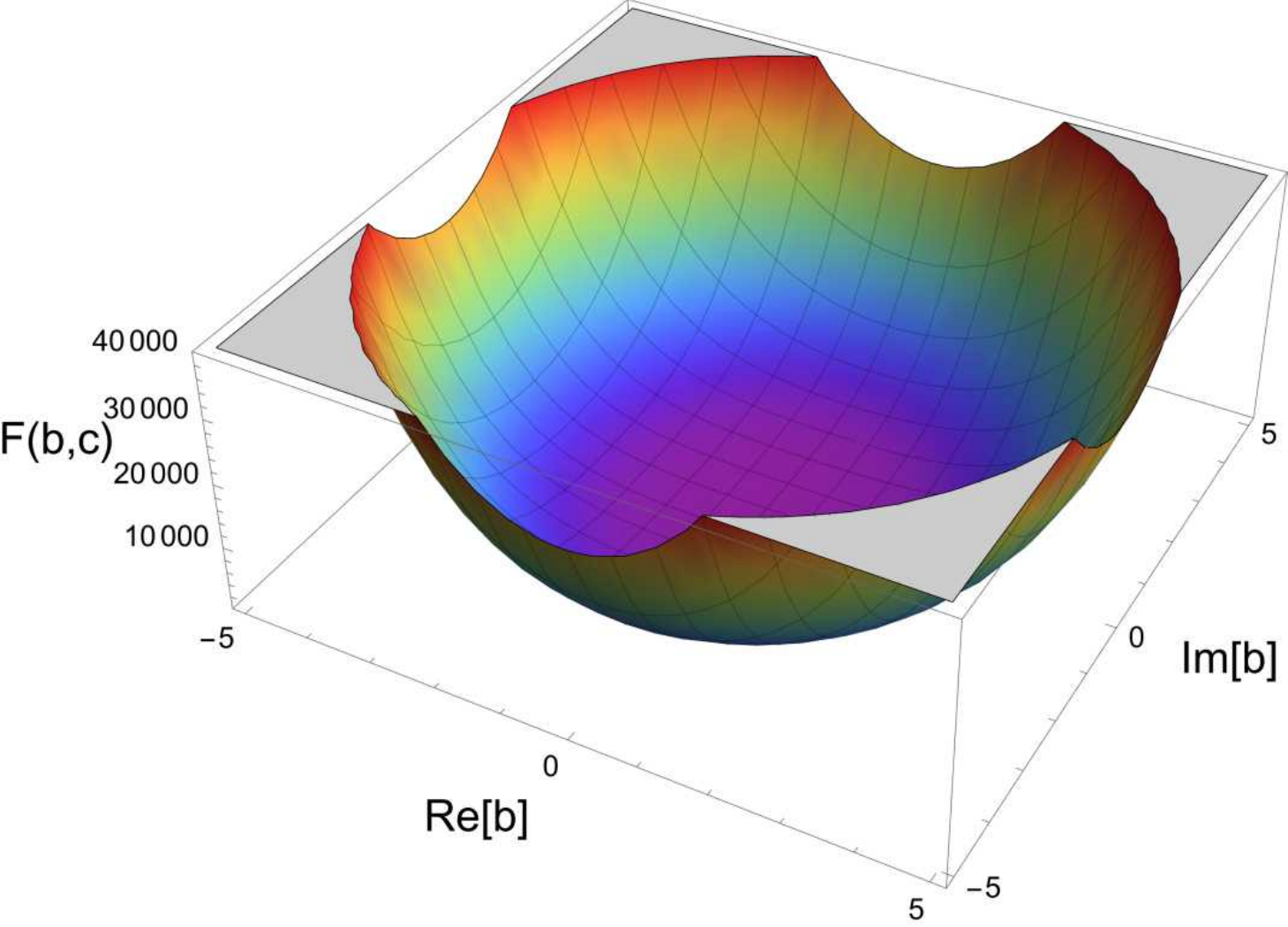}
        \caption{$c=-1-i$}
        \label{fig:a251}
    \end{subfigure}
    \hfill
    \begin{subfigure}{0.45\textwidth}
        \includegraphics[width=\textwidth]{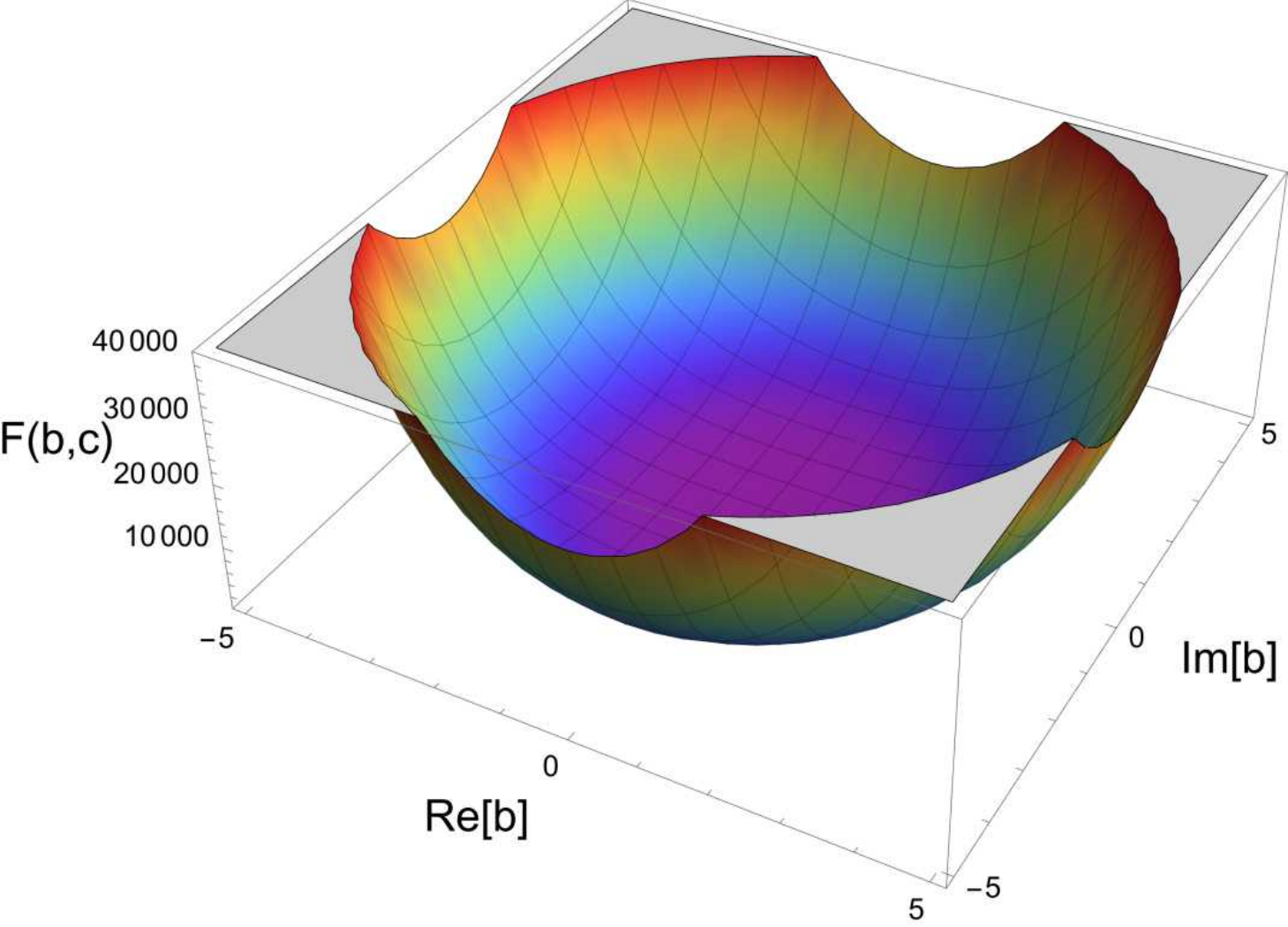}
        \caption{$c=1+i$}
        \label{fig:a252}
    \end{subfigure}
    
    \vspace{0.5cm}
    \begin{subfigure}{0.45\textwidth}
        \includegraphics[width=\textwidth]{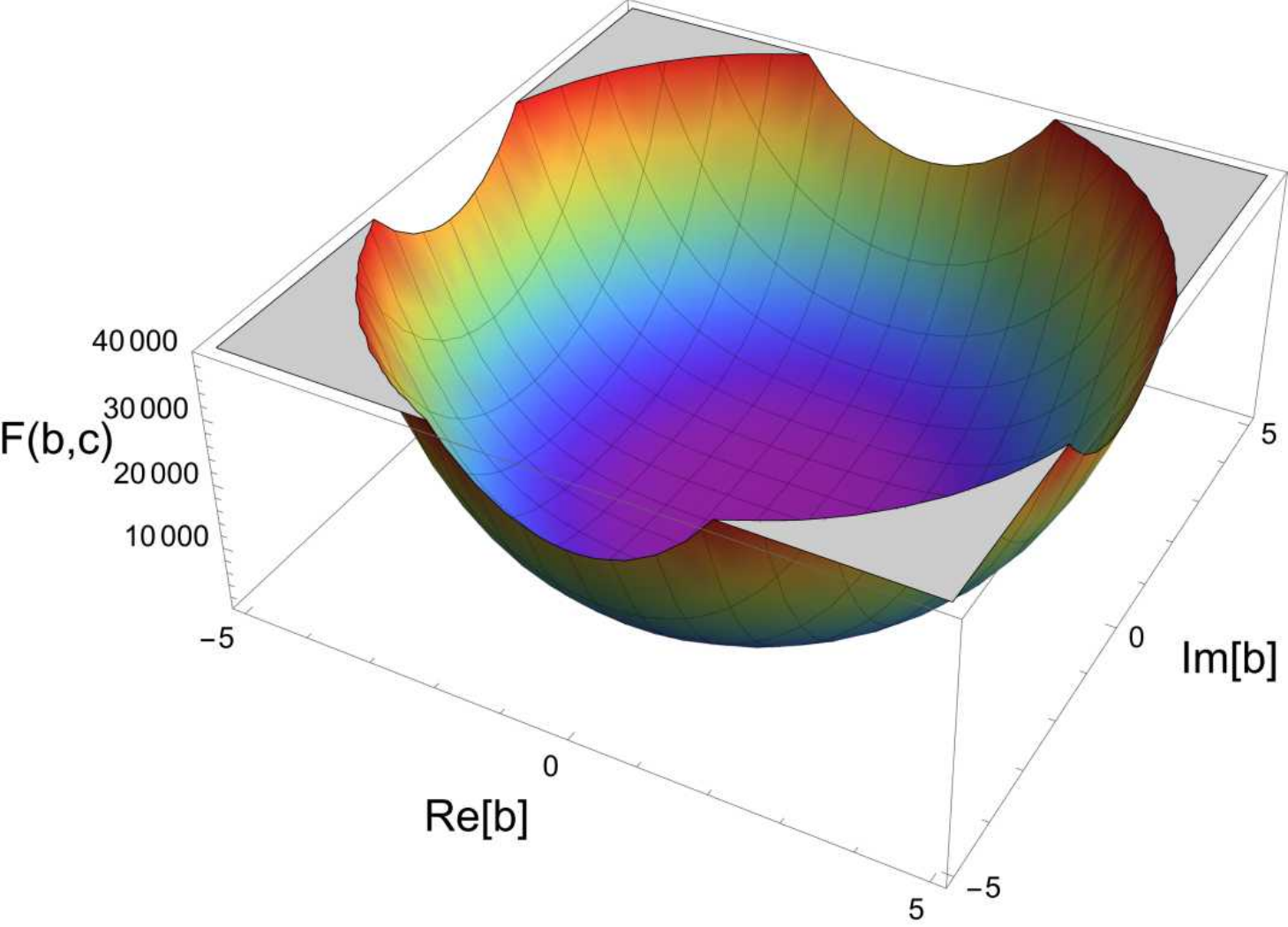}
        \caption{$c=-1+i$}
        \label{fig:a253}
    \end{subfigure}
    \hfill
    \begin{subfigure}{0.45\textwidth}
        \includegraphics[width=\textwidth]{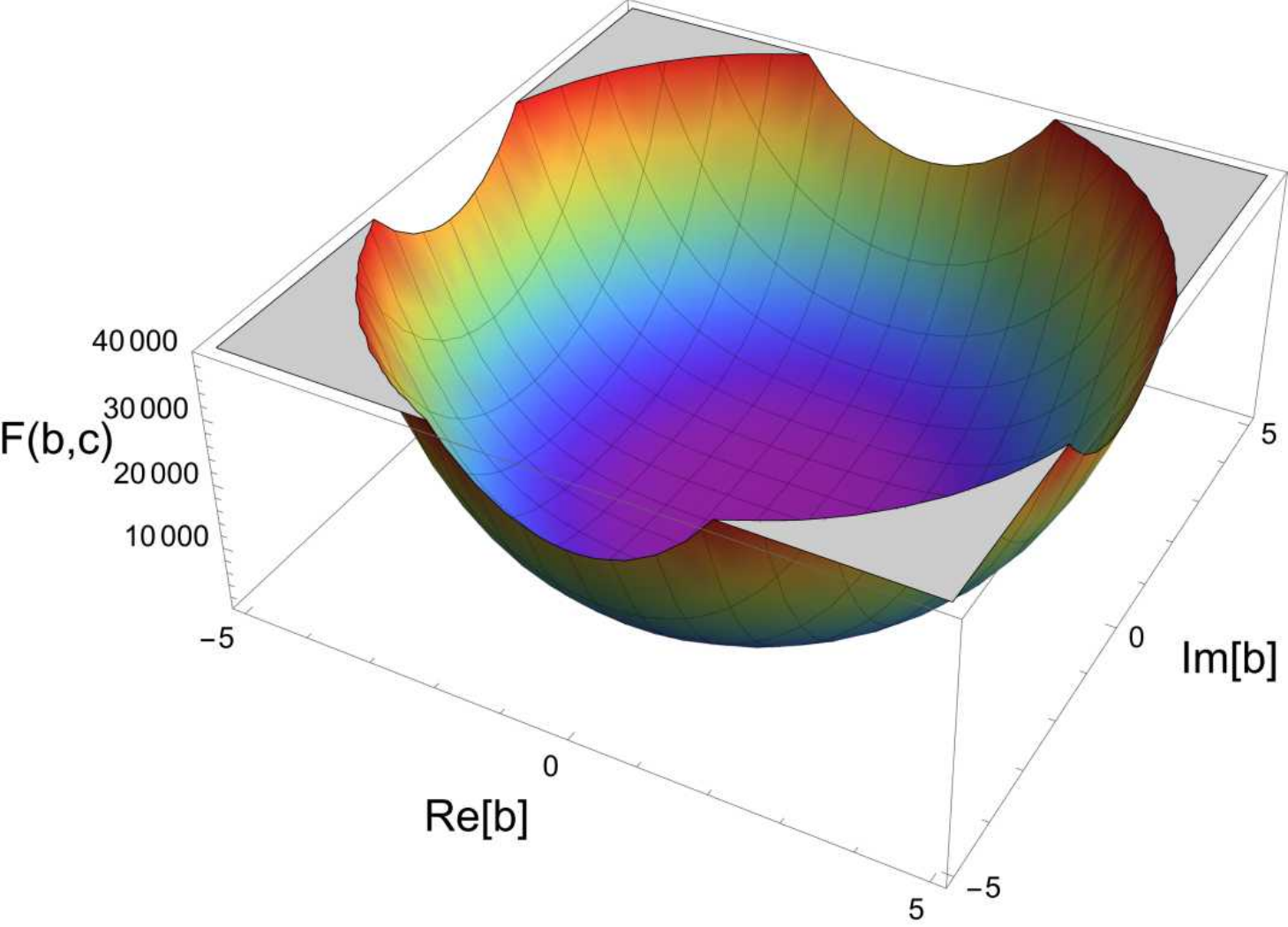}
        \caption{$c=1-i$}
        \label{fig:a254}
    \end{subfigure}
    
    \caption{$\mathrm{z-axis}:F(b,c)=1075648\a_{2}^{5}$, starts at $1$; $\mathrm{x-axis}:\mathrm{Re}[b]$; $\mathrm{y-axis}:\mathrm{Im}[b]$}
    \label{fig:a25}
\end{figure}

\begin{figure}[H]
    \centering
    \begin{subfigure}{0.9\textwidth}
        \includegraphics[width=\textwidth]{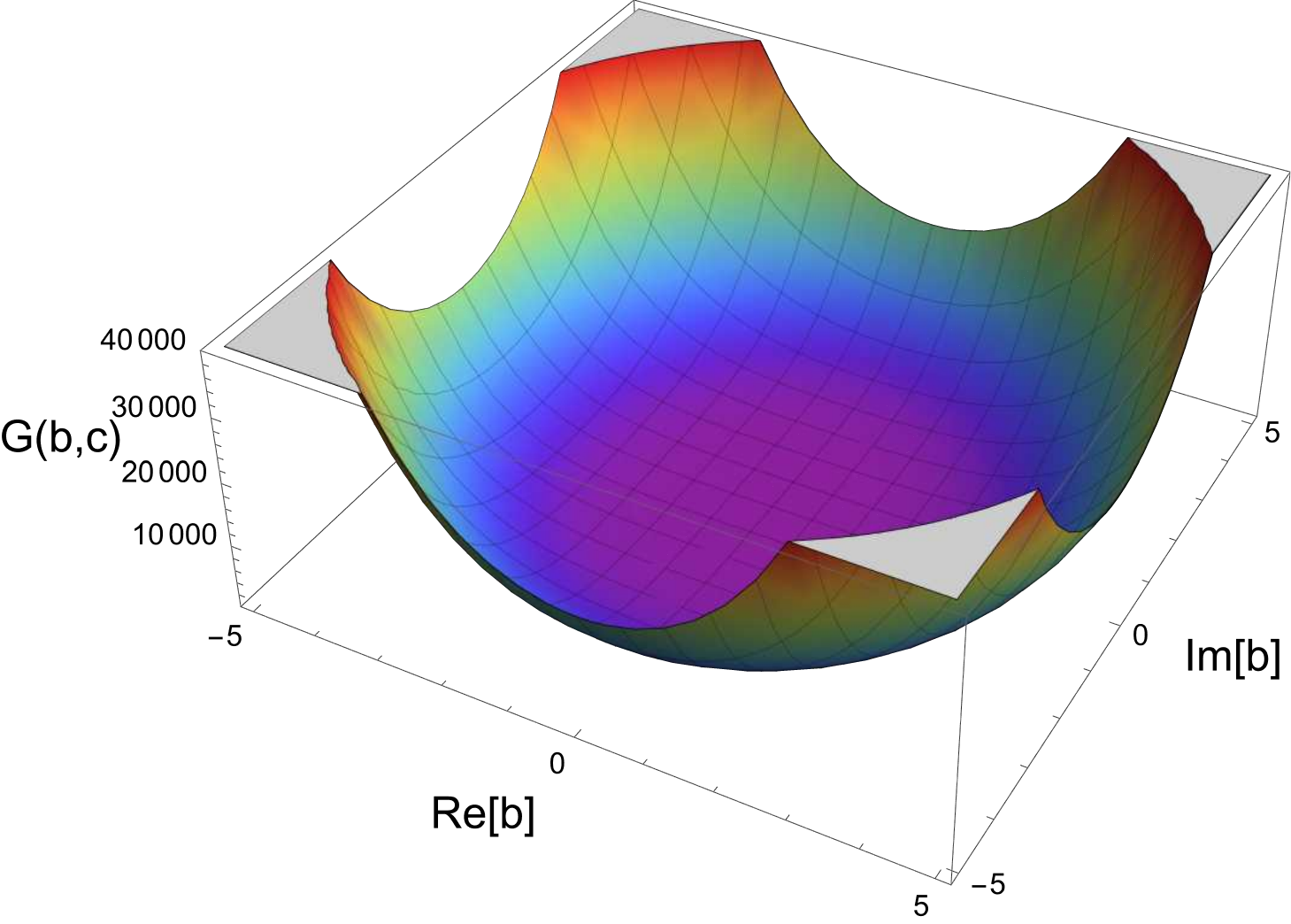}
        \caption{$c=0$}
        \label{fig:a260}
    \end{subfigure}

    \vspace{0.5cm}
    \begin{subfigure}{0.45\textwidth}
        \includegraphics[width=\textwidth]{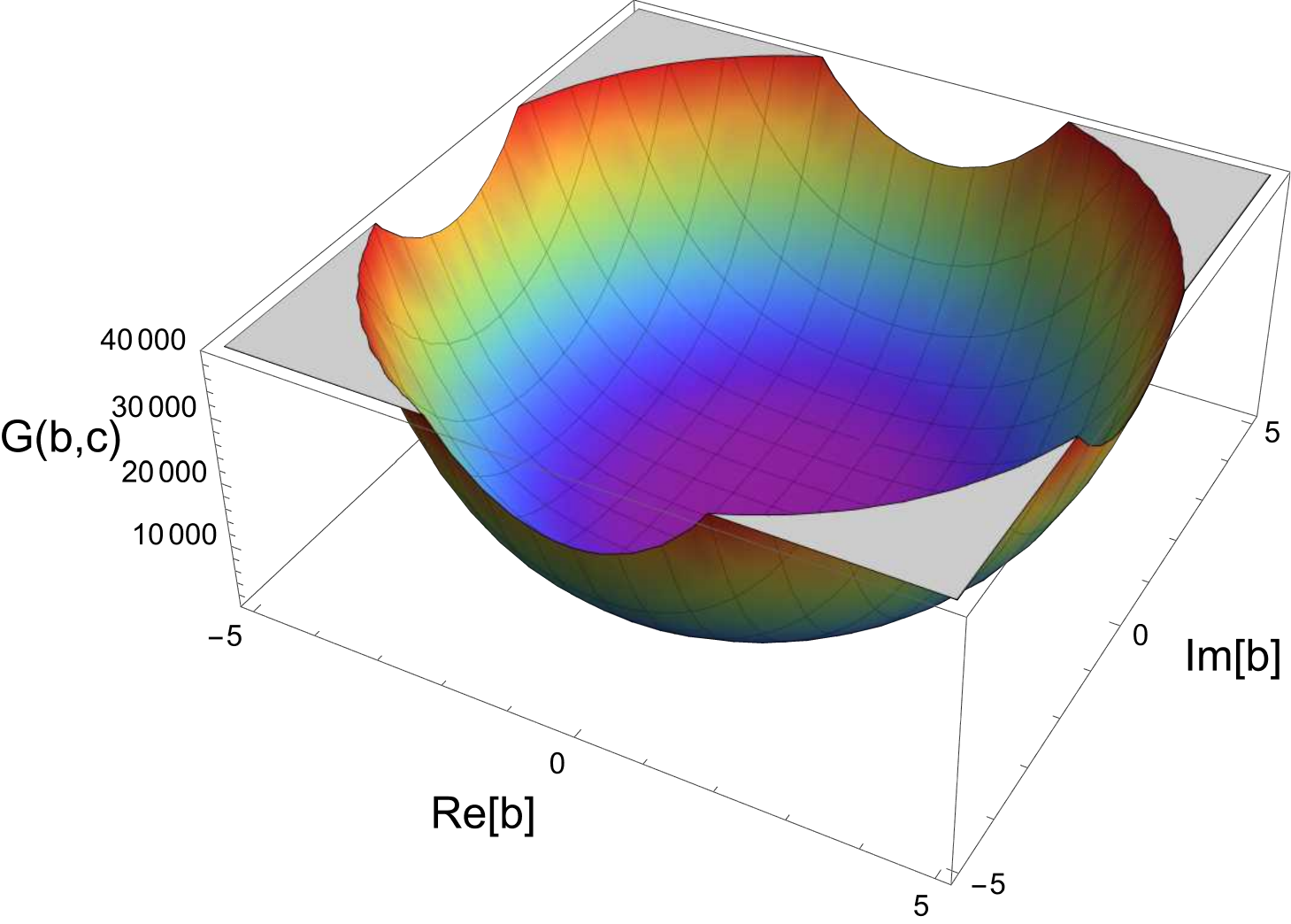}
        \caption{$c=-1-i$}
        \label{fig:a261}
    \end{subfigure}
    \hfill
    \begin{subfigure}{0.45\textwidth}
        \includegraphics[width=\textwidth]{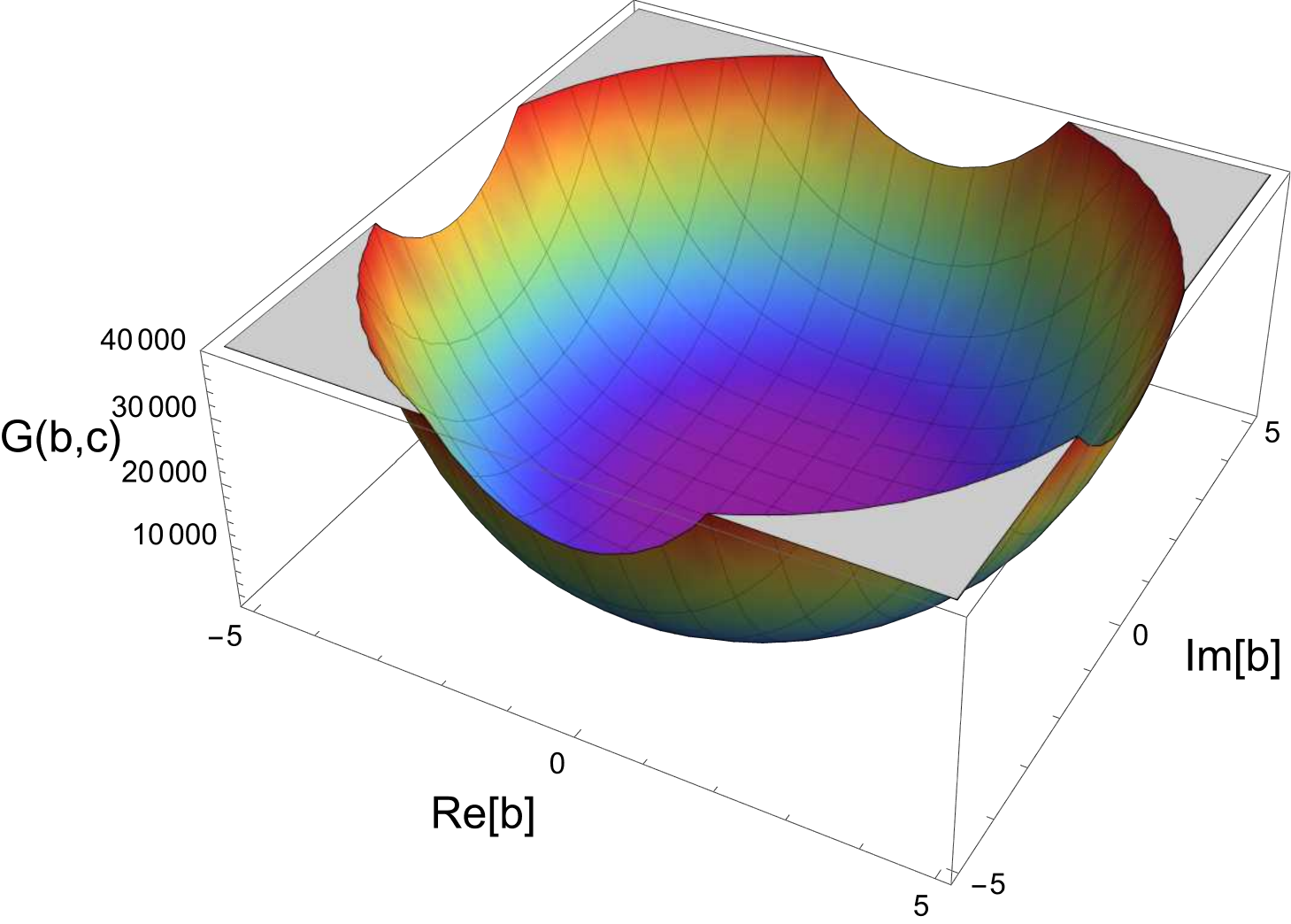}
        \caption{$c=1+i$}
        \label{fig:a262}
    \end{subfigure}
    
    \vspace{0.5cm}
    \begin{subfigure}{0.45\textwidth}
        \includegraphics[width=\textwidth]{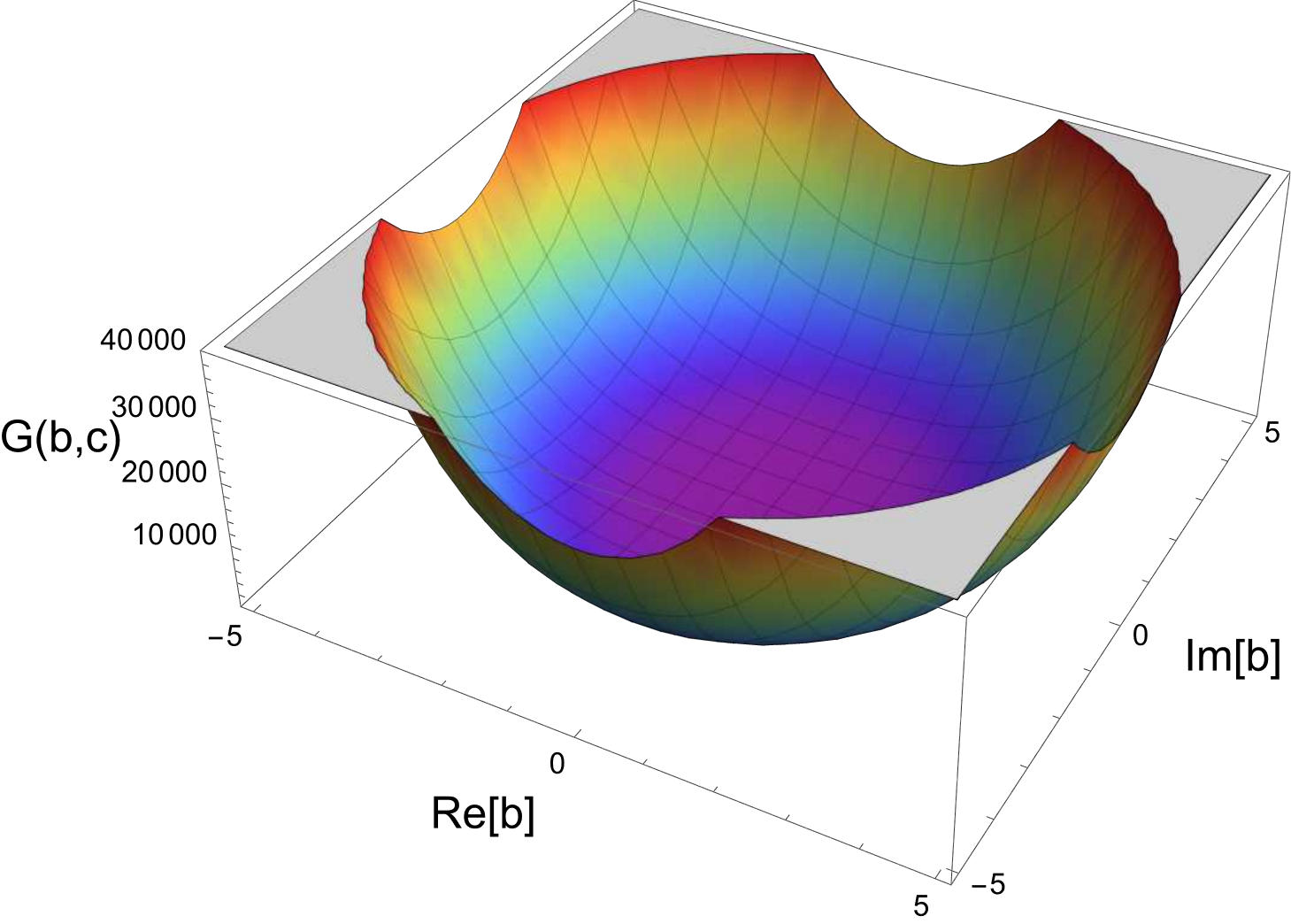}
        \caption{$c=1-i$}
        \label{fig:a263}
    \end{subfigure}
    \hfill
    \begin{subfigure}{0.45\textwidth}
        \includegraphics[width=\textwidth]{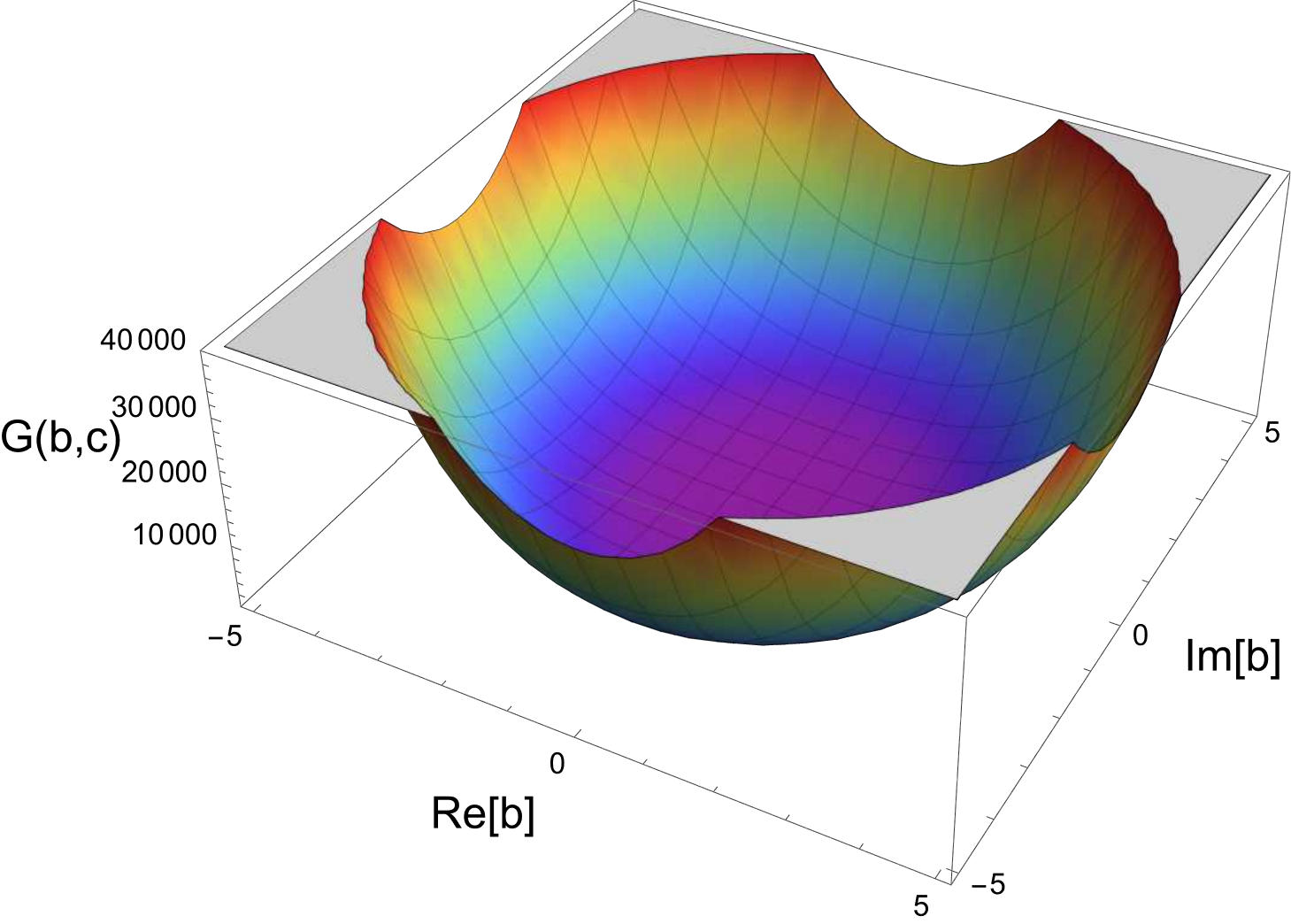}
        \caption{$c=-1+i$}
        \label{fig:a264}
    \end{subfigure}
    
    \caption{$\mathrm{z-axis}:G(b,c)=7529536\det{\a_{2}}^{x=\frac{1}{7}}$, starts at $1$; $\mathrm{x-axis}:\mathrm{Re}[b]$; $\mathrm{y-axis}:\mathrm{Im}[b]$}
    \label{fig:a26}
\end{figure}

\begin{figure}[H]
    \centering
    \includegraphics[width=0.9\linewidth]{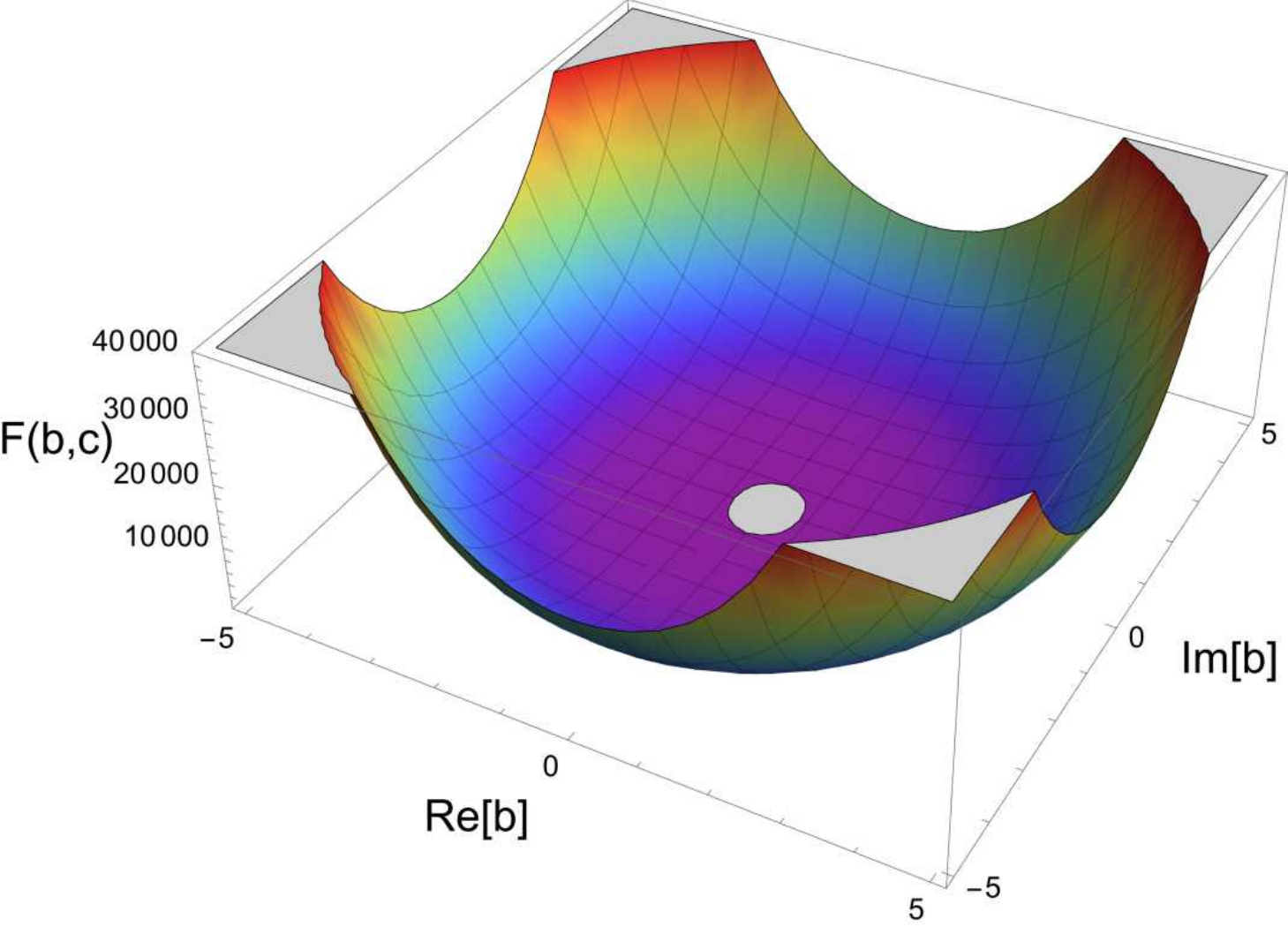}
    \caption{$\mathrm{z-axis}:F(b,c)=1075648\a_{2}^{5}$, starts at $10$; $\mathrm{x-axis}:\mathrm{Re}[b]$; $\mathrm{y-axis}:\mathrm{Im}[b]$}
    \label{fig:a255}
\end{figure}

\begin{figure}[H]
    \centering
    \includegraphics[width=0.9\linewidth]{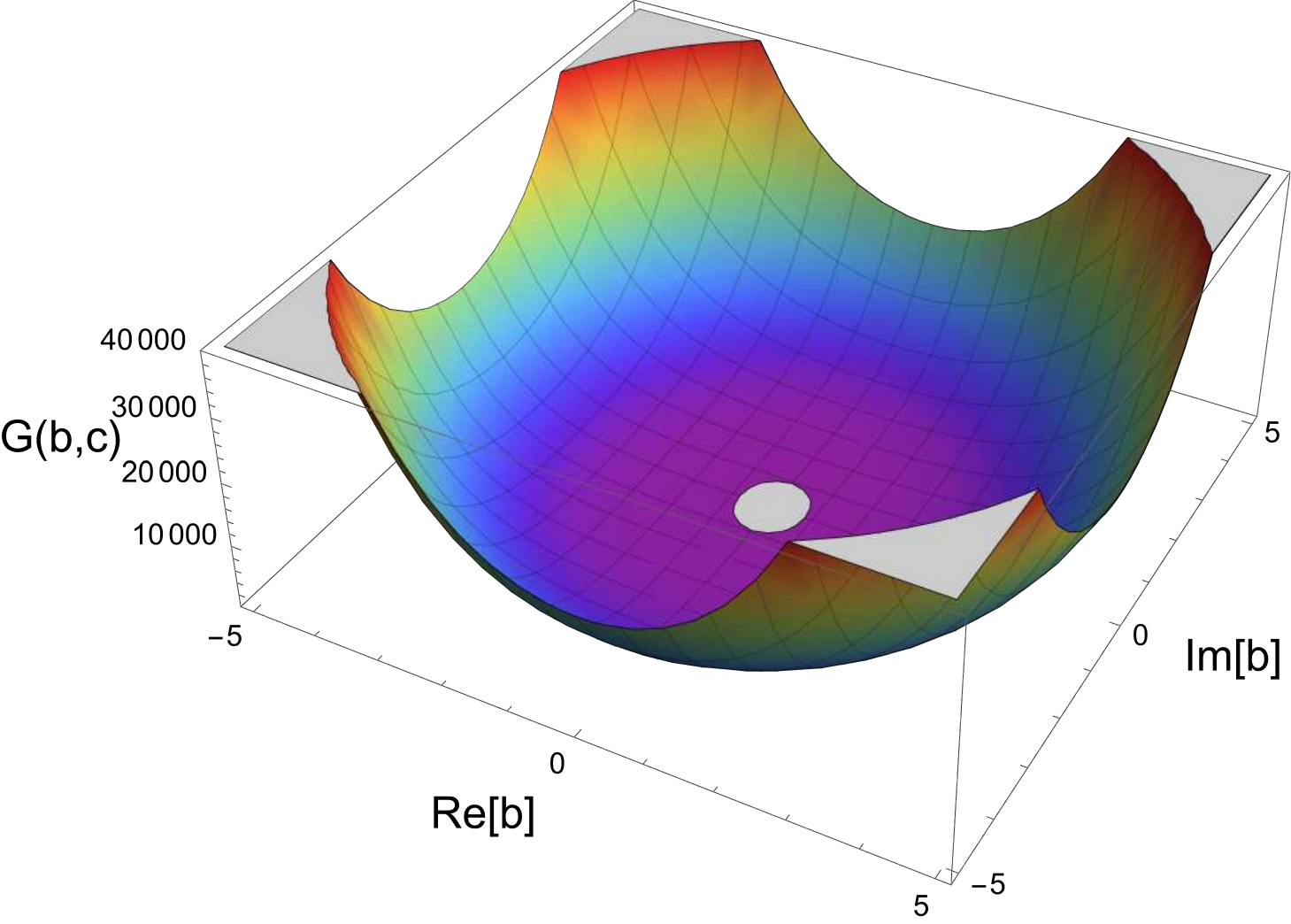}
    \caption{$\mathrm{z-axis}:G(b,c)=7529536\det{\a_{2}}^{x=\frac{1}{7}}$, starts at $10$; $\mathrm{x-axis}:\mathrm{Re}[b]$; $\mathrm{y-axis}:\mathrm{Im}[b]$}
    \label{fig:a265}
\end{figure}

\newpage

\bibliographystyle{unsrt}
\bibliography{3x3distill=zihua}

\end{document}